\documentclass[10pt,a4paper]{article}

\usepackage[latin1]{inputenc}
\usepackage{amsmath}
\usepackage{amsfonts}
\usepackage{amssymb}

\usepackage{url}

\usepackage[round]{natbib}

\usepackage{amsthm}
\usepackage{thmtools}
\theoremstyle{plain}
\newtheorem{lemma}{Lemma}
\newtheorem{proposition}{Proposition}
\newtheorem{theorem}{Theorem}

\newtheorem{corollary}{Corollary}

\theoremstyle{definition}

\usepackage{xcolor}

\usepackage{graphicx}
\usepackage{caption}
\usepackage{subcaption}
\usepackage{tikz}
\usetikzlibrary{positioning, calc, arrows.meta, chains,fit,shapes.misc}

\usepackage{booktabs}

\title{Exact and approximation algorithms for\\ the expanding search problem}

\usepackage{authblk}

\author[a]{Ben Hermans\thanks{The author is funded by a PhD Fellowship of the Research Foundation -- Flanders.}\textsuperscript{,}}
\author[a]{Roel Leus}
\author[b]{Jannik Matuschke}

\affil[a]{Research Center for Operations Research \& Business Statistics, KU~Leuven, Belgium}
\affil[b]{Research Center for Operations Management, KU Leuven, Belgium}
\date{} 
\begin{document}
	\maketitle
	\begin{abstract}
		\noindent Suppose a target is hidden in one of the vertices of an edge-weighted graph according to a known probability distribution. The expanding search problem asks for a search sequence of the vertices so as to minimize the expected time for finding the target, where the time for reaching the next vertex is determined by its distance to the region that was already searched.
		This problem has numerous applications, such as searching for hidden explosives, mining coal, and disaster relief. 
		In this paper, we develop exact algorithms and heuristics, including a branch-and-cut procedure, a greedy algorithm with a constant-factor approximation guarantee, and a novel local search procedure based on a spanning tree neighborhood.
		Computational experiments show that our branch-and-cut procedure outperforms all existing methods for general instances and both heuristics compute near-optimal solutions with little computational effort.
	\end{abstract}
	
	\section{Introduction}
	The problem of searching a hidden target plays a major role in the daily operations of search and rescue teams, intelligence services, and law enforcement agencies. Research and development departments in corporations also face a search problem when they need to identify a viable product design. One way to represent the search space is by a graph in which each vertex has a given probability to contain the target and each edge  has a length that reflects the distance (cost, time, or a composite measure) to move between the adjacent vertices. We study the problem of a searcher who, starting from a predetermined root, needs to determine how to search the graph in order to minimize the expected distance traveled before finding the target.
	
	The traditional paradigm of \emph{pathwise search} assumes that the searcher can only continue searching from their current location, which leads to a search trajectory that forms a walk in the graph. In this paper, we adopt the alternative paradigm of \emph{expanding search}, as introduced by \citet{alpern2013mining}. Here, the searcher can proceed from every previously visited vertex and the resulting trajectory forms a tree in the graph. Such a search arises naturally when the cost to traverse an edge for a second time is negligible compared to the first traversal. When mining for coal, for example, moving within existing tunnels can occur in time negligible compared to the time needed to dig into a new site \citep{alpern2013mining}. Similarly, when securing an area for a hidden explosive, moving within a secured area takes considerably less time than securing a new location \citep{angelopoulos2019expanding}.
	
	Another application for the expanding search paradigm follows from the situation where we want to minimize the weighted average (re)connection time of a graph's vertices to its root \citep{averbakh2012flowtime}. In the context of disaster relief, for example, the vertices could reflect villages, the root an aid center, and the edges road segments that have been destroyed by a natural disaster. Edge lengths represent the time necessary to clear the roads and once a road segment has been cleared, the time to travel over it becomes negligible. The problem then consists in deciding the order in which to clear the roads so as to minimize the weighted reconnection time of all villages to the aid center, where the weighting occurs according to the relative importance of different villages. This is equivalent to the expanding search problem if we interpret vertex weights as a measure of relative importance rather than as a probability that the respective vertex contains the target.

	\paragraph{Contribution and main findings.}
	In this paper, we study exact and approximation algorithms for the expanding search problem. While a substantial amount of work has been done for the pathwise search problem, research on expanding search has only emerged relatively recently. In particular, to the best of our knowledge, no approximation algorithm for the problem has been developed so far. This paper fills this gap by showing how a greedy approach yields a constant-factor approximation algorithm. As a second contribution, we develop a branch-and-cut procedure that finds an optimal expanding search sequence for instances that were not solved by  previous methods. Finally, we also describe a local search procedure.
	
	Our branch-and-cut procedure relies on a novel mixed integer program for the expanding search problem that exploits the knowledge of how to solve the special case of trees. We build upon the single-machine scheduling literature \citep{correa2005single} to obtain a linear program that identifies an optimal search sequence for a given tree, and we embed this formulation in a mixed integer program to select an optimal tree within the graph. Inspired by well-known results for the minimum spanning tree problem, we obtain valid inequalities that form the basis of our branch-and-cut procedure. Computational experiments show that our branch-and-cut procedure outperforms the existing methods for general instances and that the computational performance greatly depends on the graph's density.

	The greedy approximation algorithm starts from the intuitive idea that one prefers to probe groups of nearby vertices with a high probability to contain the object. More specifically, we employ the concept of the \emph{search density} of a subtree, defined as the ratio of (i) the total probability mass of all vertices in the subtree and (ii) the total length of all edges in the subtree. Using a proof technique inspired by \cite{feige2004approximating}, we show that a greedy approach of approximating the subtree of maximum search density and then searching that tree leads to an 8-approximation algorithm. Computational experiments suggest that, empirically, the greedy approach produces solutions much closer to optimality than the factor eight worst-case guarantee.

	Our local search procedure, finally, iteratively improves the current solution within an edge-swap neighborhood based on the spanning tree representation of the search sequence, exploiting our ability to find the optimal search for a given tree efficiently.
	We observe that it is crucial to consider all edge swaps in the metric closure of the input graph.
	We further show that this procedure finds the optimum when the underlying graph is a cycle. This contrasts with a related insertion-based local search of \cite{averbakh2012flowtime}, for which we observe that it can perform arbitrarily bad on a cycle. Moreover, our local search found the optimum for around~80\% of the instances in our computational experiment and it exceeded the optimal value by at most 2.15\% for the other cases.

	\paragraph{Related work.}
	Since the seminal work of \citet{koopman1956theory1,koopman1956theory2,koopman1957theory3}, search theory has grown to a rich area of research; see the books of \cite{alpern2003theory}, \cite{stone2007theory}, \cite{alpern2013search}, and \cite{stone2016} for general overviews. We focus on the work directly related to the expanding search problem, and compare it to previous work on the pathwise search problem.

	It has been established that the pathwise search problem is NP-hard \citep{trummel1986complexity}, even if the underlying graph forms a tree \citep{sitters2002minimum}.  \cite{archer2008faster} describe a \mbox{7.18-approximation} algorithm for the problem, and \cite{li2018multiple} give an exact solution method that also allows for multiple searchers. Many of these results exploit the connection between the pathwise search problem and the minimum latency problem, also known as the traveling repairman problem, which asks for a tour that visits all vertices and minimizes the sum of arrival times \citep{afrati1986complexity}. Indeed, the pathwise search problem can be seen as a node-weighted version of the minimum latency problem \citep{Koutsoupias1996,ausiello2000salesmen}.

	Considerably less work has been done for the expanding search problem.
	\cite{averbakh2012flowtime} show that the problem is NP-hard in general and develop exact solution methods. \cite{averbakh2012path} develops polynomial-time algorithms for a more general model that allows for multiple searchers, but only for the special case where the graph is a path. \cite{alpern2013mining} describe a polynomial algorithm for the expanding search problem when the underlying graph is a tree, and  \cite{tan2019scheduling} study both exact and approximation algorithms when there are multiple searchers, again for the special case of trees. \cite{fokkink2019submodular}, finally, generalize the algorithm of \cite{alpern2013mining} to submodular cost and supermodular weight functions. This generalization, however, does not work for general graphs since the length necessary to search a set of vertices is, in general, not a submodular function of that set.

	Several articles have also studied expanding search \emph{games}, where an adversary, or hider, chooses the target's location so as to maximize the expected search time. This then leads to a zero-sum game between the searcher and hider. In this setting, the probability distribution for the target's location cannot be taken as given, but results from an optimal (mixed) hider strategy. \cite{alpern2013mining} solve the game for trees and 2-edge-connected graphs, and \cite{alpern2019approximate} describe a method to approximate the game's value for arbitrary graphs. \cite{angelopoulos2019expanding} study a related problem where the searcher wants to minimize the expanding search ratio, defined as the maximum over all vertices of (i) the length that the search needs to reach a vertex and (ii) the shortest path between the root and that vertex. The authors show that the problem is NP-hard and find constant-factor approximation algorithms.

	\paragraph{Overview.}
	After formally defining the expanding search problem in Section~\ref{sec:prob_def}, we develop our branch-and-cut procedure in Section~\ref{sec:b&c}. Next, we describe our greedy approximation algorithm and local search in Section~\ref{sec:approx}, and report our findings from the computational experiments in Section~\ref{sec:comp_res}. Finally, we conclude and discuss potential future work in Section~\ref{sec:conclusion}.

	\section{Problem statement}\label{sec:prob_def}
	Let $G=(V,E)$ be a connected graph with root $r \in V$, a probability $p_v \in [0,1]$ associated to each vertex $v \in V$, and a length $\lambda_e \in \mathbb{R}_{> 0}$ to each edge $e \in E$.  Denote the number of non-root vertices by $n = \vert V \setminus \{r\} \vert$. We consider one searcher who has perfect information about the vertex probabilities and edge lengths. Finally, we assume that the target is hidden at exactly one of the non-root vertices such that $p_r = 0$ and $\sum_{v \in V} p_v = 1$.

	We define an \emph{expanding search} as a sequence of edges $\sigma = (e_1, \ldots, e_n)$ such that $r \in e_1$ and each edge $e_k$, with $k = 1, \ldots, n$, connects an unvisited vertex to one of the previously visited vertices. Hence, the edges $\{e_1, \ldots, e_k\}$ form a tree in $G$ for every $k = 1, \ldots, n$. Next, denote by $\lambda(\sigma,v) = \sum_{i = 1}^k \lambda_{e_{i}}$ the distance traveled until the search $\sigma = (e_1, \ldots, e_n)$ reaches vertex $v \in V$, where $e_k$ is the first edge in $\sigma$ with $v \in e_k$. The \emph{search cost} $c(\sigma)$ of an expanding search $\sigma$, defined as the expected distance traveled before finding the target, then equals
	\begin{equation}\label{eq:ex_distance}
	c(\sigma) = \sum_{v \in V} p_v \lambda(v,\sigma),
	\end{equation}
	and the \emph{expanding search problem}  asks to find an expanding search that minimizes this search cost. As mentioned above, \citet[Theorem~1]{averbakh2012flowtime} have shown that a problem equivalent to the expanding search problem is strongly NP-hard. Hence, unless $\text{P} = \text{NP}$, no polynomial-time algorithm that solves the expanding search problem exists and we need to rely on exponential-time (Section~\ref{sec:b&c}) or approximation (Section~\ref{sec:approx}) algorithms.

	\section{Branch-and-cut procedure}\label{sec:b&c}
	The mixed integer program that forms the base of our branch-and-cut procedure embeds a linear program to determine an optimal search sequence on a given tree (Section~\ref{sec:LP_trees}) within a mixed integer program that selects a tree (Section~\ref{sec:mip}).  We also introduce two classes of valid inequalities to strengthen the formulation (Section~\ref{sec:val_ineq}) and show how to separate these valid inequalities (Section~\ref{sec:sep_prob}).

	\subsection{Linear program for trees}\label{sec:LP_trees}		
	If $G$ takes the form of a rooted tree, then the searcher can only probe a vertex $v \in V \setminus \{r\}$ if all vertices on the unique path from $r$ to $v$ in the tree have already been searched. Let~$A$ denote the set of arcs obtained by directing all edges in the tree so as to reflect these precedence constraints. That is, the set~$A$ contains an arc~$(i,j)$ for each edge $\{i,j\} \in E$ such that vertex~$i$ is the immediate predecessor of vertex~$j$ in the tree's unique path from $r$ to $j$.  Several polynomial-time algorithms for the expanding search problem on a tree then arise as a special case of existing methods for sequencing with precedence constraints \citep{sidney1975decomposition,monma1979sequencing,queyranne1991single,correa2005single}.

	Define a decision variable~$\delta_{ij}$ for each pair of vertices $i,j \in V$ that indicates whether or not the search visits vertex $i$ before~$j$, and a decision variable~$z_i$ for each vertex~$i \in V$ that records the probability that the target has not been found before visiting vertex~$i$. Now consider the following linear program.
	\begin{align}
	& \text{[LP]} & & \text{min} & \sum_{(i,j) \in A} \lambda_{\{i,j\}} z_j  & & & \label{eq:lp_obj}\\
	& & & \text{s.t.} & \delta_{ij}  +  \delta_{ji} &=  1 & &\text{$\forall\, i,j \in V$,  $i \neq j$} \label{eq:lp_lo} \\
	& &  & & \delta_{ij}  +  \delta_{jk} + \delta_{ki} &\geq 1 & &\text{$\forall\, i,j,k \in V$,  $i \neq j \neq k \neq i$}\label{eq:lp_tran}  \\
	& &  & & \delta_{ij} & =  1 & & \text{$\forall\, (i,j) \in A$} \label{eq:lp_prec}\\
	& &  & & p_i + \sum_{j \in V \setminus \{i\}} p_j \delta_{ij} &= z_i   & &\text{$\forall\, i  \in V$} \label{eq:lp_prob}\\
	& &  & & \delta_{ij} & \geq 0 & &\text{$\forall\, i,j \in V$} \label{eq:lp_pos}
	\end{align}
	
	To see why this is a correct formulation for the expanding search problem on a tree, suppose first that all variables~$\delta_{ij}$ are either zero or one. Constraints~\eqref{eq:lp_lo} then make sure that vertex~$i$ is searched either before or after vertex~$j$. Together with Constraints~\eqref{eq:lp_tran} and~\eqref{eq:lp_prec}, this also enforces transitivity and the precedence constraints \citep{Potts1980}.
	Constraints~\eqref{eq:lp_prob}, in turn, state that the  target has not been found before reaching a certain vertex~$i$ with a probability equal to the total probability mass of all  vertices unvisited prior to arriving at that vertex~$i$. The objective function~\eqref{eq:lp_obj}, finally, combines the length of each arc with the probability that the searcher travels through this arc. 
	\cite{correa2005single} have shown that a linear program equivalent to~[LP], denoted \mbox{[P-LP]} in their article, has an optimal solution with all variables~$\delta_{ij}$ either zero or one. Hence, the linear program~[LP] provides a correct formulation for the expanding search problem on a tree.

	\subsection{Mixed integer program for general graphs}\label{sec:mip}
	We keep the notation of the previous section except for $A = \{(i,j), (j,i)\colon \{i,j\}\in E\}$, which now includes an arc for both possible directions to move over each edge. We say that the search uses arc $(i,j) \in A$ if the searcher travels from vertex~$i$ to vertex~$j$ through edge~$\{i,j\}$.
	In addition to the decision variables used in the linear program [LP], we introduce for each arc $(i,j) \in A$ a binary decision variable~$x_{ij}$ that indicates whether or not the search uses arc~$(i,j)$ and, similarly, a continuous decision variable~$y_{ij}$ that records the probability that the searcher did not find the target before using arc $(i,j)$. The following mixed integer program then constitutes a valid formulation for the expanding search problem on general graphs.
	\begin{align}
	& \text{[MIP]} & & \text{min} & \sum_{(i,j) \in A} \lambda_{\{i,j\}} y_{ij}  \label{eq:mip_obj}\\
	& & &\text{s.t.} & \text{\eqref{eq:lp_lo}, \eqref{eq:lp_tran}, \eqref{eq:lp_prob}, and  \eqref{eq:lp_pos}} & & &  \nonumber \\
	& & & & \sum_{i:\, (i,j) \in A} x_{ij} & =  1 & & \text{$\forall\, j \in V \setminus \{r\}$} \label{eq:mip_in}\\		
	& & & &  \sum_{i:\, (i,j) \in A} y_{ij} & =  z_j & & \text{$\forall\, j \in V \setminus \{r\}$} \label{eq:mip_prob}\\
	& & & &  0 \leq y_{ij} \leq x_{ij} & \leq \delta_{ij} & & \text{$\forall\, (i,j) \in A$} \label{eq:mip_connect}\\
	& & & &  x_{ij}  \in & \{0,1\} & & \text{$\forall\, (i,j) \in A$} \label{eq:mip_bin}
	\end{align}	
	
	Similarly as before, the objective function~\eqref{eq:mip_obj} combines the length of each edge with the probability that the searcher travels through that edge. 
	Constraints~\eqref{eq:mip_in} state that the search must reach every non-root vertex through exactly one arc, and Constraints~\eqref{eq:lp_lo}, \eqref{eq:lp_tran} and \eqref{eq:mip_connect} prevent cycles. Together with Constraints~\eqref{eq:mip_bin}, all feasible assignments to the variables $(x_{ij})_{(i,j) \in A}$  thus constitute a spanning tree of $G$ with root $r$.
	By Constraints~\eqref{eq:mip_prob}, the probability that the search reaches a vertex equals the total probability with which the search travels through an arc leading to that vertex.
	Combined with Constraints~\eqref{eq:mip_in} and~\eqref{eq:mip_connect}, this implies that for each vertex $j \in V \setminus \{r\}$ we have $y_{ij} = z_j$ if $(i,j)$ is the unique arc reaching vertex $j$ (and thus having $x_{ij} = 1$), and  $y_{ij} = 0$ otherwise. Hence, the $x$-variables select a tree, the $\delta$- and $z$-variables determine an optimal search sequence on this tree, and the $y$-variables translate this to the associated search cost.
	
	\subsection{Valid inequalities}\label{sec:val_ineq}
	The main source of weakness for the linear programming (LP) relaxation of [MIP] is that, in the presence of fractional values for the $x$-variables, there can be multiple arcs $(i,j)$ leading to the same vertex $j$ with a positive value for $y_{ij}$. The valid inequalities to be discussed in this section attempt to better coordinate the value of these $y$-variables.
	
	The first set of valid inequalities is inspired by the so-called ``directed cut model'' for spanning trees (see e.g.\ \citeauthor{magnanti1995optimal},~\citeyear{magnanti1995optimal}): 
	\begin{equation}\label{eq:flowtree_cuts}
	\sum_{(i,j) \in C(S)} y_{ij} \geq z_k \qquad \text{ for all } k \in V \setminus \{r\} \text{ and } S \subseteq V \setminus \{k\}\text{ with } r \in S.
	\end{equation}
	Here, the directed cut $C(S) = \{(i,j) \in A\colon i \in S,\, j \notin S\}$ collects all arcs starting in $S$ and ending in~$V \setminus S$. 
	The inequalities are valid because at least one arc in each directed cut $C(S)$ that separates a vertex~$k$ from the root should be traveled through with at least the probability to reach this vertex~$k$.

	Our second set of valid inequalities is similar to Constraints~\eqref{eq:flowtree_cuts}, except that now we directly use the probability parameters $(p_i)_{i \in V}$ instead of decision variables $(z_i)_{i \in V}$:
	\begin{equation}\label{eq:remprob_cuts}
	\sum_{(i,j) \in C(S)} y_{ij} \geq \sum_{i \in V \setminus S} p_i \qquad \text{ for all } S \subseteq V \text{ with } r \in S.
	\end{equation}
	These inequalities are valid since at least one arc in the directed cut~$C(S)$ should be traveled through with a probability equal to the probability mass of all vertices outside set $S$. For the special case where $S = V \setminus {j}$ for some $j \in V \setminus \{r\}$,  these inequalities can be strengthened to
	\begin{equation}\label{eq:inflow}
	z_j = \sum_{i:\, (i,j) \in A} y_{ij} \geq p_j + y_{jk}  \qquad \text{$\forall\, (j,k) \in A$.}
	\end{equation}
	Indeed, the searcher travels to vertex $j$ with a probability at least the sum of the probability that vertex~$j$ contains the object and that the searcher travels through arc~$(j,k)$.
	
	In the remainder of this paper, we refer to the two classes of Inequalities~\eqref{eq:flowtree_cuts} and~\eqref{eq:remprob_cuts}-\eqref{eq:inflow} as cuts~(C1) and~(C2), respectively.

	\subsection{Separation}\label{sec:sep_prob}
	Constraints~\eqref{eq:flowtree_cuts} and~\eqref{eq:remprob_cuts} contain exponentially many inequalities and it is therefore undesirable to include them all. Instead, we embed these inequalities in a cutting plane algorithm that iteratively solves the LP relaxation of [MIP], checks whether a violated valid inequality exists and, if so, adds it. This gives rise to a branch-and-cut procedure where we use the cutting plane algorithm to solve the LP relaxations and branch on whether or not to use an arc.

	For both Constraints~\eqref{eq:flowtree_cuts} and~\eqref{eq:remprob_cuts}, we can check whether there is a violated inequality by solving $O(n)$ minimum cut problems. Let $(x_{ij}^\star,y_{ij}^\star)_{(i,j) \in A}$, $(z_j^\star)_{j \in V}$ be a given solution to the LP relaxation and consider the directed graph $(V,A)$ in which arc $(i,j) \in A$ has capacity $y^\star_{ij}$. A violated inequality for Constraints~\eqref{eq:flowtree_cuts} then exists if and only if there is a vertex $k \in V \setminus \{r\}$ for which the minimum directed cut separating $r$ from $k$ has a capacity strictly smaller than $z^\star_k$. A similar approach allows to separate Constraints~\eqref{eq:remprob_cuts}, except that now we additionally introduce a dummy sink node $t$ in our directed graph and an extra arc $(i,t)$ for each node $i \in V$ with capacity $p_i$. Since $\sum_{i \in V \setminus S} p_i = 1 - \sum_{i \in S} p_i$ for each $S \subseteq V$, a violated inequality for Constraints~\eqref{eq:remprob_cuts} then exists if and only if the minimum directed cut has a capacity strictly smaller than one.

	\section{Approximation algorithm}\label{sec:approx}
	If the graph is a star with the root at its center, traveling to a non-root vertex~$v$ takes the same distance $\lambda_{\{r,v\}}$, independently of which vertices have been visited before. A straightforward pairwise interchange argument then shows that it is optimal to visit the vertices in non-increasing order of the ratio $p_v / \lambda_{\{r,v\}}$  \citep{smith1956various}. The \emph{search density} of a subgraph generalizes this idea to subgraphs by taking the ratio of (i) the total probability mass of all vertices in the subgraph and (ii) the total length of all edges in the subgraph. 
	
	\Citet{alpern2013mining} have shown that if $G$ is a tree, then there exists an optimal search that starts with searching the edges of a maximum density subtree. Although this approach may be suboptimal in general graphs \citep{alpern2013mining}, we will show that the resulting search cost is at most four times the optimal one. The problem of finding a subtree of maximum density in a general graph, however, is also strongly NP-hard. We develop a $1/2$-approximation for this problem (Section~\ref{sec:msd}) and show how this leads to an 8-approximation for the expanding search problem (Section~\ref{sec:descr_approx}). In Section~\ref{sec:local_search}, finally, we discuss a local search procedure to further improve the sequence found by the approximation algorithm.

	\subsection{The maximum density subtree problem}\label{sec:msd}
	Given an arbitrary graph $G^\prime$, let $V[G^\prime]$ and $E[G^\prime]$ collect the vertices and edges of~$G^\prime$, respectively.
	For a set of vertices $V^\prime \subseteq V$ of edges $E^\prime \subseteq E$, we use the notation $p(V^\prime) = \sum_{v \in V^\prime} p_v$ and $\lambda(E^\prime) = \sum_{e \in E^\prime} \lambda_e$. Let $\mathcal{T}(G)$ collect all subtrees  with root~$r$ of the graph $G$ and, given a subtree $T \in \mathcal{T}(G)$, denote $p(T) = p(V[T])$ and $\lambda(T) = \lambda(E[T])$. The \emph{search density}, or simply \emph{density}, of a subtree $T \in \mathcal{T}(G)$ with $\lambda(T) > 0$ is then defined as 
	\begin{equation*}
	\rho(T) = \frac{p(T)}{\lambda(T)},
	\end{equation*} 
	and the \emph{maximum density subtree problem} (MDSP) asks to find a tree $T^\star \in \mathcal{T}(G)$ that maximizes this density:
	\begin{equation*}	
	\rho^\star = \rho(T^\star) = \max_{T \in \mathcal{T}(G)} \rho(T).
	\end{equation*}

	The MDSP can be solved efficiently by a dynamic program in case $G$ is a tree \citep{alpern2013mining}, but it is strongly NP-hard in general \citep[Theorem~8]{lau2006finding}. \Citet{kao2013density} describe both exact and approximation algorithms for different variants of the MDSP, but, to the best of our knowledge, we are the first to develop an approximation algorithm for the MDSP as defined above.
	
	As a subroutine to our approximation algorithm, we consider a sequence of instances for the prize-collecting Steiner tree (PCST) problem. Given a connected graph $G = (V,E)$ with root $r$, edge lengths $(\lambda_e)_{e \in E}$, and vertex penalties $(p)_{v \in V}$, the PCST problem asks to find a tree $T \in \mathcal{T}(G)$ that minimizes $\lambda(T) + p(V\setminus V[T])$. \Citet{goemans1995general} have developed an approximation algorithm (henceforth called the GW algorithm) for the PCST problem with the following guarantee:   
	\begin{theorem}[\citeauthor{goemans1995general}, \citeyear{goemans1995general}]\label{th:GW}
		The GW algorithm yields a tree $T \in \mathcal{T}(G)$ with
		\begin{equation*}
		\lambda(T) + \left(2 - \frac{1}{n}\right)p(V\setminus V[T]) \leq \left(2 - \frac{1}{n}\right) (\lambda(T^\prime) + p(V\setminus V[T^\prime]))
		\end{equation*}
		for each $T^\prime \in \mathcal{T}(G)$.
	\end{theorem}

	Our approximation algorithm for the MDSP consists of a parametric search where we iteratively guess a value for the maximum density~$\rho^\star$ and evaluate our guess by employing the GW-algorithm. In particular, given a constant $\varepsilon > 0$, our \emph{parametric search} produces a subtree $T^\text{s} \in \mathcal{T}(G)$ as follows:
	
	\begin{enumerate}
		\item Take arbitrary $T^\text{s} \in \mathcal{T}(G)$, and initialize $\alpha \gets \left(2 - \frac{1}{n}\right) p(T^\text{s})/\lambda(T^\text{s})$ and $\beta \gets \max \{p_v/\lambda_{\{v,w\}} \colon \{v,w\}\in E\}$ 
		\item While $\beta > \left(1+\varepsilon \right) \alpha$, do
		\begin{enumerate}
			\item $\rho \gets (\alpha + \beta) / 2$
			\item let $T$ be the subtree obtained by applying the GW algorithm on graph~$G$ with lengths $(\rho \lambda_e)_{e\in E}$ and penalties $(p_v)_{v \in V}$
			\item if $\left(2 - \frac{1}{n}\right) p(T) \leq \rho \lambda(T)$, let $\beta \gets \rho$;
			else, let $\alpha \gets \left(2 - \frac{1}{n}\right) \rho(T)$ and $T^\text{s} \gets T$
		\end{enumerate}
		\item Return $T^\text{s}$.
	\end{enumerate}

	\begin{proposition}\label{prop:MDSP_approx}
		For each $\varepsilon > 0$, the maximum density~$\rho^\star$ is at most $\left(\left(1+\varepsilon \right)\left(2 - \frac{1}{n}\right)\right)$ times the density $\rho(T^s)$ of the tree $T^s$ obtained by the parametric search.
	\end{proposition}
	\begin{proof}
	Denote by~$\alpha_i$,~$\beta_i$, and~$T^\text{s}_i$ the values for~$\alpha$,~$\beta$, and~$T^\text{s}$ at the beginning of iteration~$i$ in the while loop, and let $T^\star$ be a tree with maximum density~$\rho^\star$ . We claim that for each iteration~$i\geq 1$ it holds that $\rho^\star \leq \beta_i$ and $\alpha_i = \left(2 - \frac{1}{n}\right) \rho(T_i^\text{s})$. Since $\beta \leq \left(1+\varepsilon \right) \alpha$ after the final iteration, this implies
	\begin{equation*}
	\rho^\star \leq \left(1+\varepsilon \right) \alpha = \left(1+\varepsilon \right) \left(2 - \frac{1}{n}\right) \rho(T^\text{s}),
	\end{equation*}
	which would then prove the result.
	
	We prove the claim by induction. By definition, $\alpha_1 = \left(2 - \frac{1}{n}\right) \rho(T_1^\text{s})$. Since for every two sequences $a_1, \ldots, a_k$ and $b_1, \ldots, b_k$ of respectively non-negative and positive numbers it holds that 
	\begin{equation*}
	\max_{j = 1, \ldots, k} \frac{a_j}{b_j} \geq \frac{\sum_{j = 1}^k a_j}{\sum_{j = 1}^k b_j},
	\end{equation*}
	we also have that $\beta_1 \geq \rho^\star$. 
	
	Next, take an arbitrary iteration~$i \geq 1$, and assume that $\beta_i \geq \rho^\star $ and $\alpha_i = \left(2 - \frac{1}{n}\right) \rho(T_i^\text{s})$. Let~$T$ be the subtree obtained by the GW algorithm and distinguish between two cases. Firstly, if
	\begin{equation*}
	\left(2 - \frac{1}{n}\right) p(T) \leq \rho \lambda(T),
	\end{equation*}
	then adding $\left(2 - \frac{1}{n}\right)p(V \setminus V[T])$ to both sides of this inequality yields that
	\begin{equation*}
	\left(2 - \frac{1}{n}\right) p(V) \leq \rho \lambda(T) + \left(2 - \frac{1}{n}\right) p(V \setminus V[T]).
	\end{equation*}
	Hence, by Theorem~\ref{th:GW}, we have for each $T^\prime \in \mathcal{T}(G)$ that
	\begin{equation*}
	\left(2 - \frac{1}{n}\right) p(V) \leq \left(2 - \frac{1}{n}\right) \left(\rho \lambda(T^\prime) +  p(V \setminus V[T^\prime])\right)
	\end{equation*}
	or, equivalently, that $\rho\lambda(T^\prime) \geq p(T^\prime)$. Thus, our guess~$\rho$ for the density was too high, and we obtain that $\beta_{i + 1} = \rho \geq \rho^\star$, while $\alpha_{i + 1} = \alpha_i$ and $T_{i + 1}^\text{s} = T_{i}^\text{s}$.
	Alternatively, if $\left(2 - \frac{1}{n}\right) p(T) > \rho \lambda(T)$, then
	\begin{equation*}
	\left(2 - \frac{1}{n}\right) p(V) > \rho \lambda(T) + \left(2 - \frac{1}{n}\right) p(V \setminus V[T]),
	\end{equation*}
	which yields that $\left(2 - \frac{1}{n}\right) \rho(T) > \rho$. Thus, our guess~$\rho$ for the density was too low, and we have $\alpha_{i + 1} = \left(2 - \frac{1}{n}\right) \rho(T_{i+1}^\text{s}) > \rho $ and $\beta_{i + 1} = \beta_{i}$. Together with the induction hypothesis, we obtain in both cases that $\rho^\star \leq \beta_{i+1}$ and $\alpha_{i+1} = \left(2 - \frac{1}{n}\right) \rho(T_{i+1}^\text{s})$, which proves the claim.
	\end{proof}
	
	As summarized by the next result, by choosing $\varepsilon = 1/(2n - 1)$, we obtain a $1/2$-approximation algorithm for the MDSP.\@ The associated running time consists of the one for the GW algorithm multiplied with the number of iterations in the while loop. In particular, \citet[Theorem~28]{hegde2015nearly} discuss how the GW algorithm can be implemented to run in time $O(\vert E \vert \log(n)\log(M))$  and, for every $\varepsilon > 0$, the number of iterations in the while loop is $O(\log(M/\varepsilon))$.
	\begin{corollary}\label{col:MDSP_half}
		For $\varepsilon = 1/(2n - 1)$, the parametric search is a $1/2$-approximation algorithm for the maximum density subtree problem that runs in time\\ $O(\vert E \vert\log(n)\log(M)\log(nM))$, with $M$ the largest input number required to describe the instance. 
	\end{corollary}

	\subsection{Greedy algorithm for the expanding search problem}\label{sec:descr_approx}
	To represent the graph structure that remains after having searched a set of vertices, we use the concept of vertex contraction.
	Given a graph $G = (V,E)$ and a subset $S \subseteq V$ with $r \in S$, denote the graph obtained by contracting the vertices in $S$ to the root by $G/S$. For each vertex $w \in V \setminus S$, let $C(w,S) = \{\{v,w\} \in E\colon v\in S\}$ collect all edges in $E$ connecting vertex~$w$ to a vertex in $S$. We then have that $V[G/S] = (V \setminus S) \cup \{r\}$, and that $E[G/S]$ consists of all edges $\{v,w\} \in E$ with $v,w \in V \setminus S$ and all edges $\{r,w\}$ with $w \in V \setminus S$ and non-empty $C(w,S)$.
	For each edge $e \in E[G/S]$ in the contracted graph, we define a length $\lambda^S_e$ as follows. If the edge is not adjacent to the root, i.e.\ if $r \notin e$, then we take the original edge length $\lambda^S_e = \lambda_e$. If the edge is of the form $e = \{r, w\}$ for some $w \in V \setminus S$, in turn, then we set $\lambda^S_e$ equal to the minimal length  to reach vertex~$w$ from the set~$S$ in the original graph, i.e.\ $\lambda^S_e = \min_{e^\prime \in C(w,S)} \lambda_{e^\prime}$.

	Our \emph{greedy algorithm} assumes that an approximation algorithm for the MDSP is available and uses this to produce a search sequence $\sigma^\text{g}$ as follows:
	\begin{enumerate}
		\item Initialize $i\gets 1$ and $S \gets \{r\}$
		\item While a vertex $v \in V \setminus S$ with $p_v > 0$ remains, do
		\begin{enumerate}
			\item let $T_i$ be the subtree obtained by applying the approximation algorithm for the MDSP on graph $G/S$ with probabilities $(p_v)_{v \in V[G/S]}$ and lengths $(\lambda^S_e)_{e \in E[G/S]}$
			\item let $\sigma_i$ be an arbitrary expanding search on tree $T_i$
			\item increment $i$ and let $S \gets S \cup V[T_i]$
		\end{enumerate}
		\item Let $\sigma^\text{g}$ be the expanding search implied by $(\sigma_1, \sigma_2, \ldots)$, where every edge adjacent to the root in the contracted graph is replaced with a corresponding length-minimizing edge in the original graph.
	\end{enumerate}

	Suppose we have an algorithm for the MDSP that can find a subtree with a density at most~$\alpha \geq 1$ times smaller than the maximum density. Employing this algorithm within our greedy procedure then leads to a search sequence with a cost at most $4\alpha$ times the optimal search cost. The proof follows a similar structure as the one of \citet[Theorem~4]{feige2004approximating}.
	
	\begin{theorem}\label{th:greedy}
		Given a $1/\alpha$-approximation algorithm for the maximum density subtree problem, the greedy algorithm is a $4\alpha$-approximation algorithm for the expanding search problem.
	\end{theorem}
	\begin{proof}
	Let $\sigma^\star$ be an optimal expanding search sequence, $\sigma^\text{g}$ the sequence obtained by the greedy algorithm, and $m$ the number of iterations in this greedy algorithm. For each iteration~$i$, let $R_i = \bigcup_{j = i}^{m} V[T_j] \setminus \{r\}$ collect the unvisited vertices prior to iteration $i$, denote by $\lambda^{(i)} = \lambda^{V \setminus R_i}(T_i)$ the total length of the tree obtained in iteration $i$, and call $\varphi_i = p(R_i)\lambda^{(i)}/p(T_i)$ the \emph{price} of iteration~$i$.
	
	To show that $c(\sigma^\text{g}) \leq 4\alpha c(\sigma^\star)$, we proceed in three steps. Firstly, Lemma~\ref{lem:greedybound} gives an upper bound on the greedy algorithm's search cost in terms of the weighted sum of prices. Next, Lemma~\ref{lem:densitybound} shows that if the searcher has already visited a set $S$ of vertices and if $\bar{\rho}$ is an upper bound on the density in the remaining graph $G/S$, then no tree in the original graph can search more than $\tau\bar{\rho}$ probability mass of the unvisited vertices within length~$\tau$. Lemma~\ref{lem:geometric}, finally, uses this result together with a geometric argument to show that the upper bound for the greedy algorithm's search cost is at most four times the optimal search cost.

	\begin{lemma}\label{lem:greedybound}
		$c(\sigma^\text{g}) \leq \sum_{i = 1}^m p(T_i) \varphi_i.$
	\end{lemma}
	\begin{proof}
	By Equation~\eqref{eq:ex_distance} and the construction of $\sigma^\text{g}$,
	\begin{equation*}
	c(\sigma^\text{g}) = \sum_{v \in V} p_v \lambda(v, \sigma^\text{g}) = \sum_{i = 1}^m \sum_{v \in V[T_i]} p_v \lambda(v, \sigma^\text{g}) \leq \sum_{i = 1}^m \sum_{v \in V[T_i]} p_v \sum_{j=1}^i\lambda^{(j)}.
	\end{equation*}
	Rearranging terms then yields that
	\begin{equation*}
	c(\sigma^\text{g}) \leq \sum_{i=1}^m \lambda^{(i)} \sum_{j=i}^m p(T_i) = \sum_{i = 1}^m \lambda^{(i)} p(R_i),
	\end{equation*}
	and the result follows since the definition of $\varphi_i$ implies that $\lambda^{(i)} p(R_i) = p(T_i)\varphi_i$.
	\end{proof}

	\begin{lemma}\label{lem:densitybound}
		Given set of vertices $S \subset V $, let $\bar{\rho}$ be such that $\bar{\rho} \geq p(T)/\lambda^S(T)$ for every $T \in \mathcal{T}(G/S)$. For each $\tau > 0$ and $T \in \mathcal{T}(G)$ with $\lambda(T) \leq \tau$ then holds that $p(V[T] \setminus S) \leq \tau \bar{\rho}$.
	\end{lemma}
	\begin{proof}
	For arbitrary $\tau > 0$ and $T \in \mathcal{T}(G)$ with $\lambda(T) \leq \tau$, consider the graph $T/S$ obtained from~$T$ by contracting the vertices of $S$ to the root. Since $T/S \in \mathcal{T}(G/S)$, the definition of $\bar{\rho}$ together with the observation that $\lambda^S(T/S) \leq \lambda(T) \leq \tau$ yields that
	\begin{equation*}
	p(V[T]\setminus S) = p(T/S) \leq \bar{\rho} \lambda^S(T/S) \leq  \tau \bar{\rho}. \qedhere
	\end{equation*}
	\end{proof}

	\begin{lemma}\label{lem:geometric}
		$\sum_{i = 1}^m p(T_i) \varphi_i \leq 4\alpha c(\sigma^\star).$
	\end{lemma}
	\begin{proof}
	The proof relies on a geometric argument and proceeds in two steps. 
	Firstly, we construct two diagrams whose surfaces equal $c(\sigma^\star)$ and $\sum_{i = 1}^m p(T_i) \varphi_i$, respectively. Next, we show that if we shrink the latter diagram by factor $4\alpha$, then it fits within the former. Figure~\ref{fig:approx} illustrates this idea.

	\colorlet{optcol}{gray!60!white}
	\colorlet{greedycol}{lightgray!60!white}
	
	\pgfmathsetlengthmacro\punit{7pt}
	\pgfmathsetlengthmacro\lunit{0.85987*\punit}
	
	\pgfmathsetlengthmacro\pmax{16*\punit}
	\pgfmathsetlengthmacro\lmax{11.5*\lunit}

	\pgfmathsetlengthmacro\pa{3*\punit}
	\pgfmathsetlengthmacro\la{2*\lunit}
	
	\pgfmathsetlengthmacro\pag{2*\punit}
	\pgfmathsetlengthmacro\prag{7.5*\lunit}
	
	\pgfmathsetlengthmacro\pb{\pa + 2*\punit}
	\pgfmathsetlengthmacro\lb{\la + 1*\lunit}
	
	\pgfmathsetlengthmacro\pbg{\pag + 9*\punit}
	\pgfmathsetlengthmacro\prbg{8.67*\lunit}
	
	\pgfmathsetlengthmacro\pc{\pb + 6*\punit}
	\pgfmathsetlengthmacro\lc{\lb + 4*\lunit}
	
	\pgfmathsetlengthmacro\pcg{\pbg + 4*\punit}
	\pgfmathsetlengthmacro\prcg{3*\lunit}
	
	\pgfmathsetlengthmacro\pd{\pc + 4*\punit}
	\pgfmathsetlengthmacro\ld{\lc + 3*\lunit}
	
	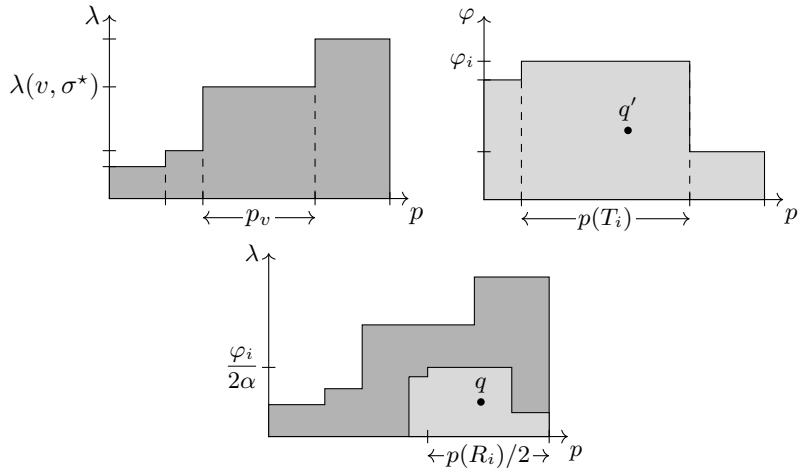
\begin{figure}
		\caption{Geometric argument to show that $\sum_{i = 1}^m p(T_i) \varphi_i \leq 4\alpha c(\sigma^\star)$.}\label{fig:approx}
		\centering
		\begin{tikzpicture}[font=\mdseries]
		\draw[->] (-0.2pt,0) -- (\pmax,0);
		\draw[->] (0,-0.2pt) -- (0,\lmax);	
		\node at (\pmax,0) [anchor=north, xshift = 3pt]{$p$};
		\node at (0,\lmax) [anchor=east]{$\lambda$};
		\draw[fill=optcol] (0, \la) -- (\pa, \la) -- (\pa, \lb) -- (\pb, \lb) -- (\pb, \lc) -- (\pc, \lc) -- (\pc, \ld) -- (\pd, \ld) -- (\pd, 0) -- (0,0) -- (0, \la);	
		\draw (\pa,2.5pt) -- (\pa, -2.5pt);
		\draw (\pb,2.5pt) -- (\pb, -2.5pt);
		\draw (\pc,2.5pt) -- (\pc, -2.5pt);
		\draw (\pd,2.5pt) -- (\pd, -2.5pt);			
		\draw (-2.5pt, \la) -- (2.5pt, \la);
		\draw (-2.5pt, \lb) -- (2.5pt, \lb);
		\draw (-2.5pt, \lc) -- (2.5pt, \lc);
		\draw (-2.5pt, \ld) -- (2.5pt, \ld);				
		\draw[dashed] (\pa, 0) -- (\pa, \la);
		\draw[dashed] (\pb, 0) -- (\pb, \lb);
		\draw[dashed] (\pc, 0) -- (\pc, \lc);
		\draw[<->] 	(\pb,-1.2*\lunit) -- node [fill=white, inner sep=1pt] {$p_{v}$} (\pc,-1.2*\lunit);
		\node[anchor = east] at (0,\lc) {$\lambda(v,\sigma^\star)$};
		\end{tikzpicture}
		\begin{tikzpicture}[font=\small]
		\draw[->] (-0.2pt,0) -- (\pmax,0);
		\draw[->] (0,-0.2pt) -- (0,\lmax);	
		\node at (\pmax,0) [anchor=north, xshift = 3pt]{$p$};
		\node at (0,\lmax) [anchor=east]{$\varphi$};
		\draw[fill=greedycol] (0, \prag) -- (\pag, \prag) -- (\pag, \prbg) -- (\pbg, \prbg) -- (\pbg, \prcg) -- (\pcg, \prcg) -- (\pcg, 0) -- (0, 0) -- (0, \prag);			
		\draw (\pag,2.5pt) -- (\pag, -2.5pt);
		\draw (\pbg,2.5pt) -- (\pbg, -2.5pt);
		\draw (\pcg,2.5pt) -- (\pcg, -2.5pt);
		\draw (-2.5pt, \prag) -- (2.5pt, \prag);
		\draw (-2.5pt, \prbg) -- (2.5pt, \prbg);
		\draw (-2.5pt, \prcg) -- (2.5pt, \prcg);				
		\draw[dashed] (\pag, 0) -- (\pag, \prag);
		\draw[dashed] (\pbg, 0) -- (\pbg, \prbg);
		\draw[<->] 	(\pag,-1.2*\lunit) -- node [fill=white, inner sep=1pt] {$p(T_i)$} (\pbg,-1.2*\lunit);
		\node[anchor = east] at (0,\prbg) {$\varphi_i$};
		\node[fill=black, inner sep=1pt, circle] at (0.7*\pbg, 0.5*\prbg) {};
		\node[anchor = south] at (0.7*\pbg, 0.5*\prbg) {$q^\prime$};
		\end{tikzpicture}		
		\begin{tikzpicture}[font=\small]
		\draw[->] (-0.2pt,0) -- (\pmax,0);
		\draw[->] (0,-0.2pt) -- (0,\lmax);	
		\node at (\pmax,0) [anchor=north, xshift = 3pt]{$p$};
		\node at (0,\lmax) [anchor=east]{$\lambda$};
		\draw[fill=optcol] (0, \la) -- (\pa, \la) -- (\pa, \lb) -- (\pb, \lb) -- (\pb, \lc) -- (\pc, \lc) -- (\pc, \ld) -- (\pd, \ld) -- (\pd, 0) -- (0,0) -- (0, \la);
		\draw[fill=greedycol] (\pcg/2, \prag/2) -- (\pcg/2 + \pag/2, \prag/2) -- (\pcg/2 + \pag/2, \prbg/2) -- (\pcg/2 + \pbg/2, \prbg/2) -- (\pcg/2 + \pbg/2, \prcg/2) -- (\pcg/2 + \pcg/2, \prcg/2) -- (\pcg/2 + \pcg/2, 0) -- (\pcg/2, 0) -- (\pcg/2, \prag/2);
		\draw (\pcg/2 +  \pag/2,2.5pt) -- (\pcg/2 +  \pag/2, -2.5pt);
		\draw (\pcg,2.5pt) -- (\pcg, -2.5pt);
		\draw (-2.5pt, \prbg/2) -- (2.5pt, \prbg/2);
		\node[fill=black, inner sep=1pt, circle] at (\pcg/2+ 0.7*\pbg/2, 0.5*\prbg/2) {};
		\node[anchor = south] at (\pcg/2+ 0.7*\pbg/2, 0.5*\prbg/2) {$q$};
		\draw[<->] 	(\pcg/2 +  \pag/2,-1.2*\lunit) -- node [fill=white, inner sep=1pt] {$p(R_i)/2$} (\pcg/2 +  \pcg/2,-1.2*\lunit);
		\node[anchor = east] at (0,\prbg/2) {$\dfrac{\varphi_i}{2\alpha}$};
		\end{tikzpicture}
	\end{figure}

	The first diagram contains one column for each non-root vertex and these $n$ columns are ordered from left to right in the same order as the optimal expanding search visits the corresponding vertices. In particular, if vertex $v \in V$ is the $k^\text{th}$ one to be searched by $\sigma^\star$, then let~$p_v$ be the width of column $k$ and $\lambda(v,\sigma^\star)$ its height. By Equation~\eqref{eq:ex_distance}, the resulting diagram's surface equals $c(\sigma^\star)$.
	The second diagram, in turn, contains~$m$ columns that correspond to the trees selected in the greedy algorithm. Let $p(T_i)$ be the width of column $i$ and $\varphi_i$ its height, then the diagram's surface equals $ \sum_{i = 1}^m p(T_i) \varphi_i$.
	
	Now shrink the second diagram's surface by dividing the width of each column by two and the height by $2\alpha$. The original diagram's surface then equals $4\alpha$ times the shrunk one.  Next, take an arbitrary point $q^\prime$ within the second diagram, align the shrunk diagram to the right of the first diagram, and call $q$ the point corresponding to $q^\prime$ within the shrunk diagram (see Figure~\ref{fig:approx}). To complete the proof, we show that $q$ lies within the first diagram.
	
	Call~$i$ the index of the column that contains $q^\prime$ in the second diagram, then the height of point~$q$ in the shrunk diagram is at most $\varphi_i/2\alpha$ and its distance from the right-hand side boundary is at most $p(R_i)/2$. Hence, it suffices to show that after having traveled length $\varphi_i/2\alpha$, at least a probability mass of $p(R_i)/2$ still needs to be searched under the optimal sequence~$\sigma^\star$. To prove this, we apply Lemma~\ref{lem:densitybound} with $S = V \setminus R_i$ and $\bar{\rho} = \alpha p(T_i)/\lambda^{(i)}$, where $\bar{\rho}$ upper bounds the density in $G/S_i$ since we obtained $T_i$ from our $1/\alpha$-approximation algorithm for the MDSP.  Lemma~\ref{lem:densitybound} then yields that in length $\tau = \varphi_i/2\alpha$, the optimal sequence $\sigma^\star$ can search at most a probability mass
	\begin{equation*}
	\frac{\varphi_i}{2\alpha}  \frac{\alpha p(T_i)}{\lambda^{(i)}} = \frac{p(R_i)}{2}
	\end{equation*}
	within the set $R_i$. Hence, after length $\varphi_i/2\alpha$,  a probability mass of at least $p(R_i)/2$ still needs to be searched, which implies that point $q$ lies within the first diagram.
	\end{proof}

	In sum, Lemmas~\ref{lem:greedybound} and~\ref{lem:geometric} yield
	\begin{equation*}
	c(\sigma^\text{g}) \leq \sum_{i = 1}^m p(T_i) \varphi_i \leq 4\alpha c(\sigma^\star),
	\end{equation*}
	which completes the proof of Theorem~\ref{th:greedy}.
	\end{proof}

	Since the greedy algorithm needs at most $n$ iterations to construct the sequence $\sigma^\text{g}$, Corollary~\ref{col:MDSP_half} and Theorem~\ref{th:greedy} yield the following result.
	\begin{corollary}\label{col:eight_approx}
		The greedy algorithm and parametric search with $\varepsilon = 1/(2n - 1)$ yield an $8$-approximation algorithm for the expanding search problem that runs in time $O(\vert E \vert n\log(n)\log(M)\log(nM))$, with $M$ the largest input number required to describe the instance. 
	\end{corollary}

	\subsection{Local search}\label{sec:local_search}
	We now describe a local search algorithm that attempts to improve a given search sequence, such as the one obtained by our greedy algorithm. 
	In this procedure, a solution is represented by a spanning tree in the metric closure of the instance graph. The associated cost is the search cost of an optimal sequence within this tree.
	A local move consists of adding one edge and removing another edge such that we obtain a new spanning tree. If the optimal search cost of the new tree is better than that of the previous tree, we proceed with it; otherwise, we try the next local move. The algorithm ends when no move exists that improves upon the current tree's search cost.

	To formalize this, denote by $\bar{G} = (V, \bar{E})$ the metric closure of graph $G = (V,E)$, i.e.\ the complete graph on $V$ with the length of edge $\{v,w\} \in \bar{E}$ equal to the shortest path between vertices~$v$ and~$w$ in~$G$. Let $T \in \mathcal{T}(\bar{G})$ be a given spanning tree of~$\bar{G}$, and define $c^\star(T)$ as the optimal search cost associated with this tree. Note that~$c^\star(T)$ can be efficiently computed using, for example, the algorithm of \cite{monma1979sequencing}. Next, call $C_{T,e}$ the \emph{fundamental cycle} of~$T$ with respect to an edge $e \in \bar{E} \setminus E[T]$, i.e.\ the unique cycle in the graph obtained by adding edge~$e$ to tree~$T$. Finally, let $T + e - e^\prime$ denote the tree obtained by adding an edge $e \in \bar{E} \setminus E[T]$ to~$T$ and removing another edge~$e^\prime \in C_{T,e}$. Our \emph{local search algorithm} then produces a feasible expanding search sequence~$\sigma^{\text{ls}}$ for the original graph~$G$ as follows:
	\begin{enumerate}
		\item Given a spanning tree $T \in \mathcal{T}(\bar{G})$, initialize $S \gets \bar{E} \setminus E[T]$
		\item While $S$ is non-empty, do
		\begin{enumerate}
			\item take an arbitrary edge~$e \in S$
			\item if $c^\star(T + e - e^\prime) < c^\star(T)$ for an edge~$e^\prime \in C_{T,e}$, let $T \gets T + e - e^\prime$ and $S \gets \bar{E} \setminus E[T]$
			\item else, let $S \gets S \setminus \{e\}$
		\end{enumerate}
		\item Let~$\sigma^{\text{ls}}$ be the search sequence in~$G$ obtained from the optimal sequence associated with~$T$ in~$\bar{G}$ by replacing each transitive edge by the underlying shortest path in~$G$ and, to avoid cycles, omitting those edges that connect two previously visited vertices.
	\end{enumerate}
	
	Our computational experiments suggest that, empirically, this local search provides a high-quality solution as it attained the optimum in around 80\% of the cases. The next lemma shows that the procedure always finds a globally optimal solution if the input graph is a cycle. To keep the proof simple, we assume that $p_v > 0$ for all $v \in V \setminus \{r\}$, which is without loss of generality.

	\begin{lemma}
		If $G$ is a cycle, every locally optimal tree is also globally optimal.
	\end{lemma}
	\begin{proof}
	Let $T \in \mathcal{T}(\bar{G})$ be a locally optimal tree in the metric closure of $G$.
	By contradiction assume that $T$ contains an edge $e = \{v, w\}$ that is in $\bar{G}$ but not in $G$.
	W.l.o.g.\ assume $e$ is the last such edge in the optimal search sequence $\sigma$ for $T$ and that $v$ is visited before $w$.
	Because $e$ is not in $G$, it corresponds to a path in $G$ containing at least one internal vertex. Let $z$ be the internal vertex of this path that is visited first in the sequence $\sigma$.
	If $z$ is visited before $w$ in $\sigma$, then the tree $T' = T + \{z, w\} - \{v, w\}$ has lower optimal search cost than $T$ (note that a feasible search sequence for $T'$ is $\sigma$ with $\{v, w\}$ replaced by $\{z, w\}$, with $\lambda_{\{z, w\}} < \lambda_{\{v, w\}}$).
	If $z$ is visited after $w$ in $\sigma$, then by choice of $e$ and $z$, either the edge $\{v, z\}$ or the edge $\{z, w\}$ must be in $T$. Note that replacing $\{v, w\} \in T$ by the respective missing edge is a feasible local move, creating a new tree $T''$. A feasible search sequence for $T''$ is $\sigma$ with $\{v, w\}$ replaced by the tuple $\{v, z\}, \{z, w\}$. In this sequence, $z$ is visited strictly earlier while no other vertex is visited later than in the optimal order for $T$. Hence, in either case, $T$ was not locally optimal, yielding a contradiction.
	
	We conclude that $T$ only contains edges of $G$.
	Because $G$ is a cycle, any tree in $G$ can be attained from $T$ by a local move.
	Because there is a tree $T^*$ in $G$ that corresponds to a globally optimal search sequence, also $T$ must be globally optimal.
	\end{proof}

	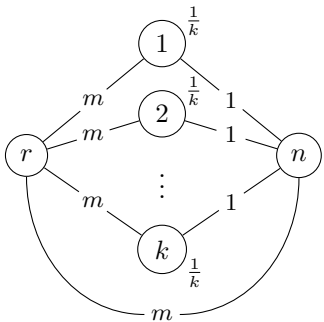
\begin{figure}
		\centering
		\caption{Instance that illustrates the importance of taking the graph's metric closure in our local search.}\label{fig:mcl}
		\begin{tikzpicture}[
		baseline,
		vertexst/.style 2 args ={circle, draw=black, fill=white, label={[inner sep = 0]{#1}: \small{#2}}},
		length/.style={fill=white, circle, inner sep=1mm},
		dots/.style={fill=white, minimum size=3.5ex},
		dot/.style={fill=black, circle, minimum size=1.5pt, inner sep=0},
		locopt/.style={},
		length/.style={fill=white, anchor=center, circle, inner sep=0.5mm, font=\small},
		]
		\node[vertexst={}{}] (r) {$r$};
		\coordinate[xshift=1.5cm] (ngroup) at (r.east);
		\node[yshift=-2ex] (dots) at (ngroup) {$\vdots$};
		\node[vertexst={20}{$\frac{1}{k}$}, anchor=south, yshift=1ex] (2) at (dots.north) {$2$};
		\node[vertexst={20}{$\frac{1}{k}$}, anchor=south, yshift=2ex] (1) at (2.north) {$1$};
		\node[vertexst={-20}{$\frac{1}{k}$}, anchor=north, yshift=-1.5ex] (k) at (dots.south) {$k$};
		\node[vertexst={}{}, xshift=1.5cm, anchor=west] (n) at (ngroup) {$n$};
		\node[length, yshift=-2ex, anchor=north] at (k.south) (m) {$m$};
		\draw[locopt] (r) -- node[length] {$m$} (1);
		\draw[locopt] (r) -- node[length] {$m$} (2);
		\draw[locopt] (r) -- node[length] {$m$} (k);
		\draw (1) -- node[length] {$1$} (n);
		\draw (2) -- node[length] {$1$} (n);
		\draw (k) -- node[length] {$1$} (n);
		\draw[locopt] (r) to[out=270, in=180] (m.west);
		\draw[locopt] (m.east) to[out=0, in=270] (n);
		\end{tikzpicture}
	\end{figure}
	
	Figure~\ref{fig:mcl} illustrates the importance of taking the metric closure of the graph. Given positive integers~$n$,~$m$, and~$k=n-1$, consider the search sequence $\sigma_1 = (\{r,1\}, \{r,2\}, \ldots, \{r,k\}, \{r,n\})$,  and the search sequence $\sigma^\star = (\{r,n\}, \{1,n\},$ $\{2,n\}, \ldots, \{k,n\})$, which is a global optimum for $m, k \geq 3$. The associated search costs are $c(\sigma_1) = \sum_{i = 1}^k \frac{im}{k} = \frac{m(k+1)}{2}$ and $c(\sigma^\star) = m + \frac{(k+1)}{2}$, such that the ratio $c(\sigma_1) / c(\sigma^\star)$ can be made arbitrarily high by choosing~$k$ and~$m$ sufficiently large. Without considering the metric closure of the graph, the tree defined by the first sequence forms a local optimum because replacing an edge~$\{r,i\}$ by another edge~$\{i,n\}$ only increases the distance to reach vertex~$i \in \{1,\ldots,k\}$. If we do consider the metric closure, however, this tree is locally optimal only if $m \leq 2$ because, otherwise, replacing edge~$\{r,2\}$ by the transitive edge~$\{1,2\}$ of length~$2$ would decrease the search cost. Since for $m = 2$ the ratio $c(\sigma_1) / c(\sigma^\star)$ tends to~$2$ as~$k$ increases, we conclude that the worst-case ratio for the instance of Figure~\ref{fig:mcl} is arbitrarily high without considering the graph's metric closure and is bounded by~$2$ if we do consider the metric closure. This latter observation also implies that the worst-case ratio of our local search is at least~$2$.

	\cite{averbakh2012flowtime} have described a different local search for expanding search. In their method a solution is represented by the sequence in which the vertices are visited instead of by the underlying spanning tree. Given a permutation $\pi = (r, \pi_1, \ldots, \pi_n)$ of the vertices in~$V$ with $\pi_k$ indicating the~$k^\text{th}$ vertex to be searched, define a spanning tree $T^\pi \in \mathcal{T}(\bar{G})$ by including for each $k \in \{1,\ldots,n\}$ a minimum-length edge in the metric closure~$\bar{G}$ that connects vertex~$\pi_k$ to one of the vertices in~$\{r,\pi_1,\ldots,\pi_{k-1}\}$. Given a permutation~$\pi$, the local search of \cite{averbakh2012flowtime} then checks whether there exist~$i$ and~$j$, with $i < j$, such that the permutation~$\pi^\prime$ obtained by inserting vertex~$\pi_j$ at position~$i$ in permutation~$\pi$ and shifting all vertices $(\pi_i,\ldots, \pi_{j-1})$ one position backwards attains a lower search cost $c^\star(T^{\pi^\prime}) < c^\star(T^\pi)$. If so, the procedure repeats with permutation~$\pi^\prime$, otherwise~$\pi$ is a local optimum. Observe that \cite{averbakh2012flowtime} focus on complete graphs in their article, but that this is equivalent to our setting when taking the metric closure.

	\pgfmathsetlengthmacro\radius{2.2cm}
	\begin{figure}
		\centering
		\caption{Instance for which the local search of \cite{averbakh2012flowtime} has locality gap $\Omega(n)$.}\label{fig:ap_bad}
		\begin{tikzpicture}[
		baseline,
		vertexst/.style 2 args ={circle, draw=black, fill=white, label={[inner sep = 0]{#1}: \small{#2}}},
		length/.style={fill=white, circle, inner sep=0.5mm, font=\small},
		dots/.style={fill=white, minimum size=3.5ex},
		dot/.style={fill=black, circle, minimum size=1.5pt, inner sep=0},
		edgest/.style={bend left},
		]
		\draw (0,0) circle[radius=\radius];
		\node[vertexst={0}{}] (r) at ({Mod(270+0*360/7, 360)}:\radius) {$r$};
		\node[vertexst={0}{$\frac{1}{n}$}] (1) at ({Mod(270+1*360/7, 360)}:\radius) {$1$};
		\node[vertexst={0}{$\frac{1}{n}$}] (2) at ({Mod(270+2*360/7, 360)}:\radius) {$2$};
		\node[dots] (dots1) at ({Mod(270+3*360/7, 360)}:\radius) {};
		\node[vertexst={0}{}, label={[inner sep = 0, yshift=3pt]{178}: \small{$\frac{1}{n}$}}] (i) at ({Mod(270+4*360/7, 360)}:\radius) {$i$};
		\node[dots] (dots2) at ({Mod(270+5*360/7, 360)}:\radius) {};
		\node[vertexst={180}{$\frac{1}{n}$}] (n) at ({Mod(270+6*360/7, 360)}:\radius) {$n$};
		\node[dot] at ({Mod(270+3*360/7, 360)}:\radius) {};
		\node[dot] at ({Mod(270+3*360/7 + 3, 360)}:\radius) {};
		\node[dot] at ({Mod(270+3*360/7 - 3, 360)}:\radius) {};
		\node[dot] at ({Mod(270+5*360/7, 360)}:\radius) {};
		\node[dot] at ({Mod(270+5*360/7 + 3, 360)}:\radius) {};
		\node[dot] at ({Mod(270+5*360/7 - 3, 360)}:\radius) {};
		\node[length] (r1) at ({Mod(270+0.5*360/7, 360)}:\radius) {$n$};
		\node[length] (12) at ({Mod(270+1.5*360/7, 360)}:\radius) {$1$};
		\node[length] (2d) at ({Mod(270+2.5*360/7, 360)}:\radius) {$2$};
		\node[length, rectangle] (di) at ({Mod(270+3.5*360/7, 360)}:\radius) {$i-1$};
		\node[length] (id) at ({Mod(270+4.5*360/7, 360)}:\radius) {$i$};
		\node[length, rectangle] (n1r) at ({Mod(270+5.5*360/7, 360)}:\radius) {$n-1$};
		\node[length] (nr) at ({Mod(270+6.5*360/7, 360)}:\radius) {$n$};
		\end{tikzpicture}
	\end{figure}
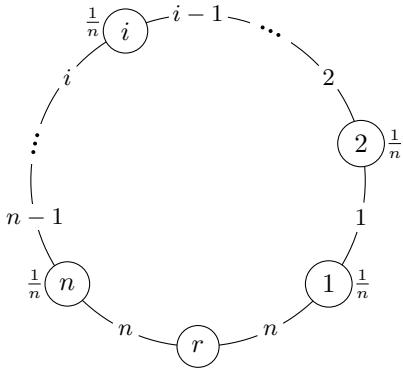

	Although the insertion-based local search of \cite{averbakh2012flowtime} is a natural approach, Figure~\ref{fig:ap_bad} illustrates that it can perform arbitrarily bad on a cycle. This contrasts with our local search procedure that, as argued above, finds a global optimum on a cycle. Referring to Figure~\ref{fig:ap_bad}, the `clockwise' sequence~$\pi = (r,n,n-1,\ldots, 1)$ forms a local optimum for the procedure of \cite{averbakh2012flowtime} with search cost $c^\star(T^\pi) = \sum_{i=1}^n \sum_{j=i}^n \frac{j}{n} = \frac{(n+1)(2n+1)}{6}$ (it is easy to check that changing only a single element in the sequence does not result in lower cost). The optimal expanding search $\sigma^\star = (\{r,1\}, \{1,2\}, \ldots, \{n-1,n\})$, however, goes counterclockwise and attains a search cost of $c(\sigma^\star) = \sum_{i=1}^n \frac{n + i - 1}{n} = \frac{3n-1}{2}$. Hence, the ratio $c^\star(T^\pi) / c(\sigma^\star)$ grows proportionally with $n$, establishing that the locality gap of this procedure is $\Omega(n)$.

	\section{Computational experiment}\label{sec:comp_res}
	This section reports how our methods perform on existing and newly generated test instances. After providing details regarding the implementation, we test our branch-and-cut algorithm on the benchmark instances provided by \cite{averbakh2012flowtime}. The results indicate that our method outperforms the current state of the art for weighted instances (i.e.\ with probabilities that differ across vertices) and that sparse graphs are considerably easier to solve than dense ones. We also provide numerical evidence that our valid inequalities and heuristics lead to reasonably tight lower and upper bounds for the optimal search cost.

	\subsection{Implementation details}\label{sec:impl_det}
	All our algorithms were implemented using the C++ programming language, compiled with Microsoft Visual C++ 14.0, and run using an Intel Core i7-4790 processor with 3.60 GHz CPU speed and 8 GB of RAM under a Windows 10 64-bit OS. To solve the mixed integer programs, we employed the commercial solver IBM ILOG CPLEX 12.8 using only one thread. Apart from the warm start provided by our greedy algorithm combined with the local search, all CPLEX parameters were set to their default values.  
	
	Unless mentioned otherwise, we ran our branch-and-cut algorithm with both classes of valid inequalities, i.e.\ with cuts~(C1) and~(C2).  Violated valid inequalities were separated using the Ford-Fulkerson algorithm \citep{ford_fulkerson_1956}. Relying on more efficient algorithms for the maximum flow problem, such as the push-relabel algorithm  of \cite{goldberg1988new}, should not affect our results because separating the cuts  requires only a relatively small fraction of the computation time. For the GW subroutine in our greedy algorithm, finally, we used the implementation of \cite{hegde2015nearly} that is available on GitHub (\url{https://github.com/fraenkel-lab/pcst_fast}). As part of our local search, we employed the polynomial-time algorithm of \cite{monma1979sequencing} to determine an optimal sequence for a given tree, and we took the tree returned by the greedy algorithm as input for the local search procedure.

	\subsection{Performance on instances of \cite{averbakh2012flowtime}}\label{sec:comp_ap}
	\cite{averbakh2012flowtime} consider both an unweighted and a weighted version of their problem, which corresponds to whether or not all vertices have equal probabilities. They further distinguish between random and euclidean network structures. In the first case, the distance between each pair of vertices is randomly generated, after which the metric closure of the graph is taken by computing shortest paths. In the second case, the distances are generated by associating with each vertex a point in the euclidean plane and then taking the euclidean distance between each pair of points and rounding to the nearest integer. For both types, the authors generated 10 benchmark instances with and without vertex weights for each of the considered network sizes.

	\begin{table}
		\centering
		\setlength\tabcolsep{0pt}
		\caption{Average CPU time in seconds for the branch-and-bound method of \cite{averbakh2012flowtime} using an Intel Core 2 Duo 2.33 GHz processor and a 30-min time limit. The number of solved instances appears in parentheses if not all 10 could be solved.} \label{tab:comp_ap}
		\begin {tabular}{r@{\extracolsep {12pt}}r@{ \extracolsep {0pt}}l@{\extracolsep {12pt}}r@{ \extracolsep {0pt}}l@{\extracolsep {12pt}}r@{ \extracolsep {0pt}}l@{\extracolsep {12pt}}r@{ \extracolsep {0pt}}l@{\extracolsep {12pt}}}%
		\toprule & \multicolumn {4}{c}{unweighted} & \multicolumn {4}{c}{weighted} \\ \cmidrule (lr){2-5} \cmidrule (lr){6-9} $n$ & \multicolumn {2}{c}{random} & \multicolumn {2}{c}{euclidean} & \multicolumn {2}{c}{random} & \multicolumn {2}{c}{euclidean} \\ \midrule %
		10&\pgfutilensuremath {0.01}&&\pgfutilensuremath {0.01}&&\pgfutilensuremath {0.00}&&\pgfutilensuremath {0.00}&\\%
		15&\pgfutilensuremath {0.01}&&\pgfutilensuremath {0.01}&&\pgfutilensuremath {0.01}&&\pgfutilensuremath {0.01}&\\%
		20&\pgfutilensuremath {0.01}&&\pgfutilensuremath {0.01}&&\pgfutilensuremath {0.02}&&\pgfutilensuremath {0.02}&\\%
		25&\pgfutilensuremath {0.01}&&\pgfutilensuremath {0.01}&&\pgfutilensuremath {1.10}&&\pgfutilensuremath {0.30}&\\%
		30&\pgfutilensuremath {0.01}&&\pgfutilensuremath {0.02}&&\pgfutilensuremath {9.59}&&\pgfutilensuremath {2.57}&\\%
		35&\pgfutilensuremath {0.07}&&\pgfutilensuremath {0.09}&&\pgfutilensuremath {229.40}&&\pgfutilensuremath {72.27}&\\%
		40&\pgfutilensuremath {0.05}&&\pgfutilensuremath {0.36}&&\pgfutilensuremath {836.94}&(4)&\pgfutilensuremath {487.96}&(7)\\%
		45&\pgfutilensuremath {0.15}&&\pgfutilensuremath {2.37}&&&&\pgfutilensuremath {1{,}377.24}&(3)\\%
		50&\pgfutilensuremath {0.91}&&\pgfutilensuremath {11.86}&&&&&\\%
		55&\pgfutilensuremath {4.71}&&\pgfutilensuremath {70.52}&&&&&\\%
		60&\pgfutilensuremath {12.57}&&\pgfutilensuremath {165.83}&(9)&&&&\\%
		65&\pgfutilensuremath {135.35}&&\pgfutilensuremath {468.80}&(8)&&&&\\%
		70&\pgfutilensuremath {263.14}&&\pgfutilensuremath {858.16}&(5)&&&&\\\bottomrule %
		\end {tabular}%
	\end{table}

	Tables~\ref{tab:comp_ap}-\ref{tab:comp_ap2} displays the average CPU time in seconds for the branch-and-bound algorithm of \cite{averbakh2012flowtime} and for our branch-and-cut method, where the averaging occurs over the solved instances within each setting. Importantly, the CPU times for the algorithm of \cite{averbakh2012flowtime} have been taken directly from their article because we were unable to obtain or to replicate their code. To partially compensate for the different processor speeds (2.33 GHz versus 3.60 GHz), we employ a time limit of 20 minutes to solve the instances instead of the 30 minutes used in \cite{averbakh2012flowtime}.  If not all 10 instances were solved within their respective time limit, the number of solved instances appears in parentheses.
	
		\begin{table}
		\centering
		\setlength\tabcolsep{0pt}
		\caption{Average CPU time in seconds for our branch-and-cut algorithm using an Intel Core i7-4790 3.60 GHz processor and a 20-min time limit. The number of solved instances appears in parentheses if not all 10 could be solved.} \label{tab:comp_ap2}
		\begin {tabular}{r@{\extracolsep {12pt}}r@{ \extracolsep {0pt}}l@{\extracolsep {12pt}}r@{ \extracolsep {0pt}}l@{\extracolsep {12pt}}r@{ \extracolsep {0pt}}l@{\extracolsep {12pt}}r@{ \extracolsep {0pt}}l@{\extracolsep {12pt}}}%
		\toprule & \multicolumn {4}{c}{unweighted} & \multicolumn {4}{c}{weighted} \\ \cmidrule (lr){2-5} \cmidrule (lr){6-9} $n$ & \multicolumn {2}{c}{random} & \multicolumn {2}{c}{euclidean} & \multicolumn {2}{c}{random} & \multicolumn {2}{c}{euclidean} \\ \midrule %
		10&\pgfutilensuremath {0.15}&&\pgfutilensuremath {0.15}&&\pgfutilensuremath {0.15}&&\pgfutilensuremath {0.16}&\\%
		15&\pgfutilensuremath {0.24}&&\pgfutilensuremath {0.32}&&\pgfutilensuremath {0.22}&&\pgfutilensuremath {0.29}&\\%
		20&\pgfutilensuremath {0.98}&&\pgfutilensuremath {1.09}&&\pgfutilensuremath {0.67}&&\pgfutilensuremath {1.09}&\\%
		25&\pgfutilensuremath {1.67}&&\pgfutilensuremath {3.31}&&\pgfutilensuremath {2.60}&&\pgfutilensuremath {1.94}&\\%
		30&\pgfutilensuremath {2.34}&&\pgfutilensuremath {7.76}&&\pgfutilensuremath {2.76}&&\pgfutilensuremath {22.15}&\\%
		35&\pgfutilensuremath {7.86}&&\pgfutilensuremath {27.83}&&\pgfutilensuremath {6.91}&&\pgfutilensuremath {57.24}&\\%
		40&\pgfutilensuremath {10.05}&&\pgfutilensuremath {60.82}&&\pgfutilensuremath {13.54}&&\pgfutilensuremath {110.67}&\\%
		45&\pgfutilensuremath {12.54}&&\pgfutilensuremath {157.80}&&\pgfutilensuremath {18.44}&&\pgfutilensuremath {144.04}&\\%
		50&\pgfutilensuremath {71.44}&&\pgfutilensuremath {192.16}&(9)&\pgfutilensuremath {48.60}&(9)&\pgfutilensuremath {161.69}&(8)\\%
		55&\pgfutilensuremath {41.69}&&\pgfutilensuremath {436.99}&(7)&\pgfutilensuremath {70.98}&&\pgfutilensuremath {168.96}&(6)\\%
		60&\pgfutilensuremath {98.79}&&\pgfutilensuremath {577.51}&(4)&\pgfutilensuremath {112.21}&(9)&\pgfutilensuremath {683.07}&(3)\\%
		65&\pgfutilensuremath {77.60}&&\pgfutilensuremath {601.35}&(3)&\pgfutilensuremath {195.91}&&\pgfutilensuremath {486.17}&(4)\\%
		70&\pgfutilensuremath {166.78}&(9)&&&\pgfutilensuremath {162.57}&&&\\\bottomrule %
		\end {tabular}%
	\end{table}	
	
	The results indicate that our branch-and-cut method outperforms the branch-and-bound algorithm of \cite{averbakh2012flowtime} for the weighted but not for the unweighted instances. This behavior occurs because the method used to compute lower bounds for the unweighted instances in \cite{averbakh2012flowtime} differs from the one used to compute lower bounds for the weighted instances. As a result, also the bounds' quality differs for the two cases. Our formulation, in contrast, remains unaffected by vertex weights. 
	
	Observe that our branch-and-cut algorithm solves random instances with considerably less effort than it needs to solve the euclidean ones. One possible explanation is that our formulation might implicitly ignore those edges whose random length was replaced by a shorter path; indeed, these edges are dominated because they have a length equal to a path that also visits intermediate vertices. Ignoring these dominated edges could improve the performance because it essentially reduces the number of spanning trees that our method needs to enumerate. The next section examines this effect in further detail.
	
	For the unsolved instances, we also recorded the relative gap between the best upper and lower bound obtained after reaching the time limit. The average \emph{optimality gap}, not included in the tables, was at most 2.16\% for our branch-and-cut method, whereas \cite{averbakh2012flowtime} reported an average optimality gap of up to 20.77\%.

	\subsection{Influence of network density}\label{sec:infl_netwd}
	The \emph{density} of a graph $G = (V,E)$, not to be confused with Section~\ref{sec:approx}'s search density, is defined as the ratio between the actual number of edges $\vert E \vert$ and the maximum possible number of edges $\vert V \vert (\vert V \vert - 1) / 2$. To examine how the network density influences the computational performance, a new set of test instances was generated in which we control for the network density. In particular, for all considered graph sizes and densities, ten instances were randomly generated as follows. An integer weight~$a_v$ between~0 and~1000 was drawn randomly for each vertex $v \in V \setminus \{r\}$, after which we set the probability $p_v = a_v / \sum_{w \in V \setminus \{r\}} a_w$. To generate edges, first a spanning tree was constructed by randomly selecting a sequence of edges and including the next edge whenever it connects two previously unconnected vertices. Secondly, edges were randomly added until the desired density was obtained. Edge lengths were generated by assigning an arbitrary point in the natural cube between~0 and~100 to each vertex and then taking the rectilinear (or Manhattan) distance between those vertices that are connected by an edge.

	\begin{table}
		\centering
		\setlength\tabcolsep{0pt}
		\caption{
			Influence of network density on average CPU time (in seconds), number of solved instances (\# out of 10), and average optimality gap (in \%) using a 20-min time limit.
		}\label{tab:comp_mip}
		\begin {tabular}{r@{\extracolsep {12pt}}r@{\extracolsep {8pt}}r@{\extracolsep {8pt}}l@{\extracolsep {12pt}}r@{\extracolsep {8pt}}r@{\extracolsep {8pt}}l@{\extracolsep {12pt}}r@{\extracolsep {8pt}}r@{\extracolsep {8pt}}l@{\extracolsep {12pt}}}%
		\toprule & \multicolumn {3}{c}{density $ = 20\%$} & \multicolumn {3}{c}{density $ = 60\%$} & \multicolumn {3}{c}{density $ = 100\%$}\\ \cmidrule (lr){2-4} \cmidrule (lr){5-7} \cmidrule (lr){8-10}$n$&\multicolumn {1}{c}{CPU}&\multicolumn {1}{c}{\#}&\multicolumn {1}{c}{gap}&\multicolumn {1}{c}{CPU}&\multicolumn {1}{c}{\#}&\multicolumn {1}{c}{gap}&\multicolumn {1}{c}{CPU}&\multicolumn {1}{c}{\#}&\multicolumn {1}{c}{gap}\\\midrule %
		\pgfutilensuremath {10}&\pgfutilensuremath {0.08}&\pgfutilensuremath {10}&\pgfutilensuremath {0.00}&\pgfutilensuremath {0.11}&\pgfutilensuremath {10}&\pgfutilensuremath {0.00}&\pgfutilensuremath {0.12}&\pgfutilensuremath {10}&\pgfutilensuremath {0.00}\\%
		\pgfutilensuremath {20}&\pgfutilensuremath {0.27}&\pgfutilensuremath {10}&\pgfutilensuremath {0.00}&\pgfutilensuremath {0.52}&\pgfutilensuremath {10}&\pgfutilensuremath {0.00}&\pgfutilensuremath {0.79}&\pgfutilensuremath {10}&\pgfutilensuremath {0.00}\\%
		\pgfutilensuremath {30}&\pgfutilensuremath {1.20}&\pgfutilensuremath {10}&\pgfutilensuremath {0.00}&\pgfutilensuremath {2.24}&\pgfutilensuremath {10}&\pgfutilensuremath {0.00}&\pgfutilensuremath {4.74}&\pgfutilensuremath {10}&\pgfutilensuremath {0.00}\\%
		\pgfutilensuremath {40}&\pgfutilensuremath {2.20}&\pgfutilensuremath {10}&\pgfutilensuremath {0.00}&\pgfutilensuremath {8.19}&\pgfutilensuremath {10}&\pgfutilensuremath {0.00}&\pgfutilensuremath {20.19}&\pgfutilensuremath {10}&\pgfutilensuremath {0.00}\\%
		\pgfutilensuremath {50}&\pgfutilensuremath {6.82}&\pgfutilensuremath {10}&\pgfutilensuremath {0.00}&\pgfutilensuremath {83.46}&\pgfutilensuremath {10}&\pgfutilensuremath {0.00}&\pgfutilensuremath {95.13}&\pgfutilensuremath {8}&\pgfutilensuremath {0.95}\\%
		\pgfutilensuremath {60}&\pgfutilensuremath {59.34}&\pgfutilensuremath {9}&\pgfutilensuremath {0.32}&\pgfutilensuremath {141.96}&\pgfutilensuremath {9}&\pgfutilensuremath {0.79}&\pgfutilensuremath {291.45}&\pgfutilensuremath {7}&\pgfutilensuremath {0.51}\\%
		\pgfutilensuremath {70}&\pgfutilensuremath {108.15}&\pgfutilensuremath {10}&\pgfutilensuremath {0.00}&\pgfutilensuremath {536.39}&\pgfutilensuremath {6}&\pgfutilensuremath {0.64}&\pgfutilensuremath {305.09}&\pgfutilensuremath {3}&\pgfutilensuremath {1.29}\\%
		\pgfutilensuremath {80}&\pgfutilensuremath {194.68}&\pgfutilensuremath {10}&\pgfutilensuremath {0.00}&\pgfutilensuremath {478.04}&\pgfutilensuremath {5}&\pgfutilensuremath {0.91}&\pgfutilensuremath {720.39}&\pgfutilensuremath {2}&\pgfutilensuremath {0.99}\\%
		\pgfutilensuremath {90}&\pgfutilensuremath {321.08}&\pgfutilensuremath {10}&\pgfutilensuremath {0.00}&\pgfutilensuremath {1{,}011.24}&\pgfutilensuremath {2}&\pgfutilensuremath {0.75}&&&\\%
		\pgfutilensuremath {100}&\pgfutilensuremath {454.05}&\pgfutilensuremath {6}&\pgfutilensuremath {0.79}&&&&&&\\%
		\pgfutilensuremath {110}&\pgfutilensuremath {676.50}&\pgfutilensuremath {4}&\pgfutilensuremath {0.38}&&&&&&\\%
		\pgfutilensuremath {120}&\pgfutilensuremath {802.89}&\pgfutilensuremath {2}&\pgfutilensuremath {0.85}&&&&&&\\\bottomrule %
		\end {tabular}%
	\end{table}

	Table~\ref{tab:comp_mip} shows that the network density strongly affects the CPU times. For graphs with $n = 40$ vertices and a density of~$100\%$, for example, the average CPU time exceeds the one for a density of~$20\%$ with almost factor ten. To a lesser extent, the average optimality gap for unsolved instances also seems to increase in the network density. As mentioned above, one possible reason why solving dense graphs requires more effort is that our method needs to enumerate more spanning trees for these graphs. Indeed, the binary $x$-variables in our mixed integer program essentially select a spanning tree, and the more such trees the graph contains, the more branching needs to be done by the branch-and-cut algorithm.

	\subsection{Quality of formulation, valid inequalities, and heuristics}\label{sec:q_bounds}
	Solving the linear programming (LP) relaxation obtained by relaxing the integrality constraints in our mixed integer program provides a lower bound on the minimal search cost. The greedy algorithm, on the other hand,  yields an upper bound as it returns a feasible starting solution. The goal of adding the valid inequalities (Section~\ref{sec:val_ineq}) and local search (Section~\ref{sec:local_search}), finally, is to tighten these bounds. To assess the quality of our formulation, valid inequalities, and heuristics, we examine how close these bounds are to the minimal search cost. Figure~\ref{fig:qual_bounds} plots the average ratio (averaged over the 10 solved instances per setting) between the bounds and optimum as a function of network size and density.
	
	\begin{figure}
		\centering
		\caption{Average ratio between bounds and the optimum. The bounds provided by the LP relaxation without cuts and with only cuts~(C2) deteriorate as network size and density increase, whereas the  other bounds remain reasonably close to the optimum.} \label{fig:qual_bounds}
		\begin{subfigure}{0.99\linewidth}
			\caption{Ratio as a function of network size $n$ for a given $60\%$ network density.}
			\centering
			\includegraphics{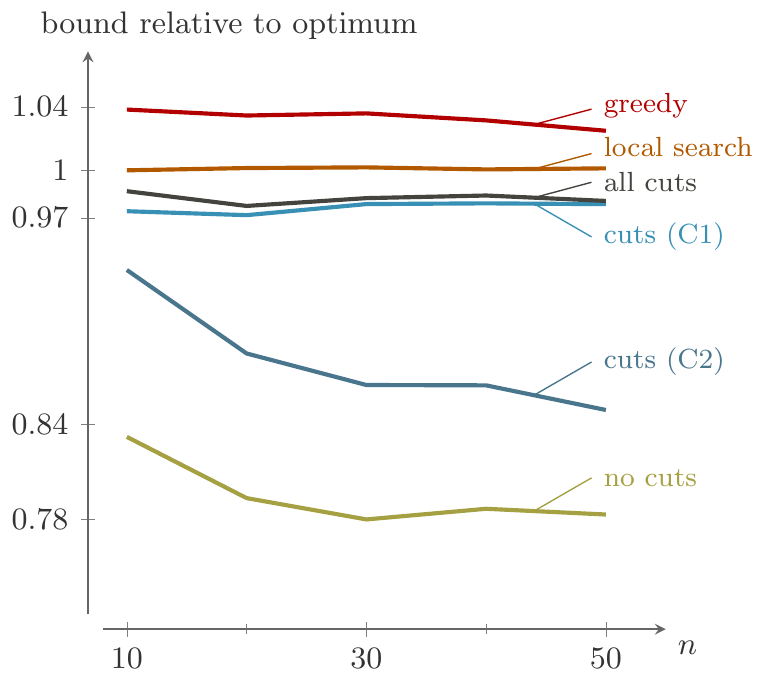}
		\end{subfigure}
		
		\begin{subfigure}{0.99\linewidth}
			\centering
			\caption{Ratio as a function of network density for a given network size $n = 40$.}
			\includegraphics{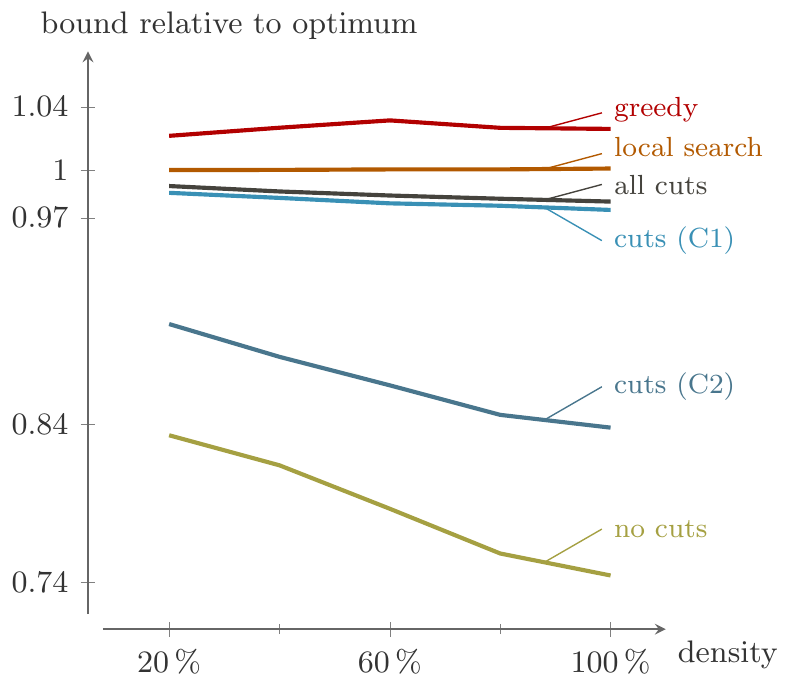}
		\end{subfigure}
	\end{figure}

	With a lower bound always exceeding $97\%$ of the minimal search cost, the LP relaxation together with both cuts~(C1) and~(C2) provides a reasonable lower bound. Although cuts~(C2) do considerably increase the lower bound compared to the formulation without valid inequalities, cuts~(C1) clearly have the most influence: the LP relaxation with only cuts~(C1) achieves a lower bound only marginally weaker than the one with all cuts included.
	
	Figure~\ref{fig:qual_bounds} also indicates that our heuristics provide high-quality solutions, especially when including the local search algorithm. On average, the greedy solution exceeds the optimal search cost with at most $4\%$, which drops to only $0.27 \%$ after performing the local search. Moreover, out of all 379 instances solved to optimality in our computational experiment, the local search could find the optimum for 305 instances and exceeded the minimal search cost with at most  $2.15\%$ for the other cases. With respect to computation times, the greedy heuristic always required less than 0.3 seconds and the local search procedure needed at most 100 seconds.  We conclude that, empirically, the greedy algorithm produces solutions much closer to optimality than the factor eight worst-case guarantee.

	Finally, observe that the lower bound provided by the LP relaxation without cuts and with only cuts~(C2) weakens as network size and density increase.
	The other bounds, in contrast, seem to be more stable. This inverse relation between network density and the quality of the formulation without cuts constitutes another reason why instances with a higher density require more effort: to make up for the weaker lower bound, cuts need to be generated. Not only does separating these cuts take time, but, more importantly, adding them also increases the time needed to solve each LP relaxation because the model grows in size.

	\section{Discussion}\label{sec:conclusion}
	We have studied exact and approximation algorithms for the expanding search problem, which  has received increasing attention in the literature since its introduction by \cite{alpern2013mining}. Our novel branch-and-cut method can solve instances that were unsolved by existing algorithms and the greedy algorithm described in this paper is the first constant-factor approximation algorithm established for the problem.
	
	The formulation at the base of our branch-and-cut procedure builds upon results from the single-machine scheduling literature and uncouples the selection of a tree from the sequencing on that tree. We believe this could be a fruitful approach for dealing with other scheduling problems as well. In particular, our formulation seems generalizable to the problem of scheduling with OR precedence constraints \citep{gillies1995scheduling}.

	Our results also contribute to the literature on search games. More specifically, consider the expanding search game where the target acts as an adversary that hides in one of the graph's vertices in order to maximize the expected search time. Together with the work of \cite{hellerstein2019solving}, our branch-and-cut procedure yields a method to compute the game's value exactly, whereas our greedy approach allows to approximate this value within factor eight. Note that these results are incomparable with the approximations of \cite{alpern2019approximate} because, in that article, the target can hide at any point of an edge in the network instead of being restricted to hide at vertices only.
	
	The tree-based local search procedure introduced in Section~\ref{sec:local_search} showed promising results in computational experiments. We also established that it finds a global optimum in case the underlying graph is a cycle, whereas other natural approaches for local neighborhoods fail on these instances. Determining the actual worst-case guarantee of the local search is an interesting problem for future research.
	
	Finally, in line with the work of \cite{tan2019scheduling}, it would be interesting to generalize the expanding search model to multiple searchers. In fact, by taking the search sequence returned by our greedy algorithm as a so-called `list schedule', the techniques of \cite{chekuri2001approximation} and \cite{tan2018schedulingWP,tan2019scheduling} seem to provide a promising starting point to obtain constant-factor approximation algorithms for the setting with multiple searchers.

	\section*{Acknowledgments}
	Ben Hermans is funded by a PhD Fellowship of the Research Foundation -- Flanders (Fonds Wetenschappelijk Onderzoek). We thank I.\ Averbakh and J.\ Pereira for kindly sharing their instances.

	\bibliographystyle{agsm}
	\bibliography{bib_ES}

\end{document}